\documentclass[final]{llncs}

\usepackage{amsmath}
\usepackage{fancyhdr}
\usepackage{relsize}
\usepackage[scr=rsfs]{mathalpha}

\iftrue
\usepackage[obeyFinal]{todonotes}
\else
\usepackage[]{todonotes}

\usepackage{catchfile}
\fi

\newcommand{\nke}[1]{\todo[inline,color=green!20!white]{#1 ---Nandakumar}}

\newcommand{\trtO}[0]{\ensuremath{0_t}}
\newcommand{\trtI}[0]{\ensuremath{1_t}}
\usepackage{algorithm}
\usepackage{algpseudocode}
\usepackage{amssymb}
\usepackage[maxbibnames=99]{biblatex}
\usepackage{booktabs}
\usepackage{comment}
\usepackage{hyperref}
\usepackage{listings}
\usepackage{graphicx}
\usepackage{pgfplots}
\usepackage{tcolorbox}

\algrenewcommand\textproc{}

\newcommand{\rocq}[0]{Rocq}
\newcommand{\ebpf}[0]{eBPF}

\usepackage[color]{coqdoc}

\newcommand{\Bit}[0]{\ensuremath{\mathbb{B}}}
\newcommand{\Lshift}[2]{\ensuremath{\mathrm{lshift}_{#1}\left(#2\right)}}
\newcommand{\Bvec}[1]{\ensuremath{\Bit^{#1}}}
\newcommand{\Trit}[0]{\ensuremath{\mathbb{T}}}
\newcommand{\tle}[0] {\ensuremath{\subseteq_t}}

\newcommand{\Tnum}[1]{\ensuremath{\Trit^{#1}}}

\newcommand{\Cdom}[0]{\ensuremath{\mathscr{C}}}

\newcommand{\bit}[2]{\ensuremath{#2_{#1}}}

\newcommand{\Case}[2]{\textbf{Case #1:} #2}
\newcommand{\basecase}[1]{\textbf{Base Case:} #1}
\newcommand{\inductionStep}{\textbf{Induction Step.}}
\newcommand{\wfdomn}{\mathbb{V}^{n}}

\newcommand{\ith}[1]{#1[i]}
\newcommand{\lsb}[1]{#1[0]}
\newcommand{\tv}[1]{#1.v}
\newcommand{\tm}[1]{#1.m}

\newcommand{\ithv}[1]{\ith{\tv{#1}}}
\newcommand{\ithm}[1]{\ith{\tm{#1}}}

\newcommand{\ingamma}[2]{#1 \in \gamma(#2)}

\newcommand{\tadd}[2]{tnum\_add(#1, #2)}
\newcommand{\tunion}[2]{tnum\_union(#1, #2)}

\newcommand{\seti}[1]{set(#1, i)}

\newcommand{\sv}[2]{(\tv{#1} + \tv{#2})}
\newcommand{\sm}[2]{(\tm{#1} + \tm{#2})}

\newcommand{\orvm}[1]{(\tv{#1} \mathbin{|} \tm{#1})}

\newcommand{\rocqlisting}[1]{\addvspace{0.5em} \coqdocnoindent #1\addvspace{0.5em}}

\renewcommand{\coqdockw}[1]{{\color{\coqdockwcolor}{\texttt{#1}}}}
\renewcommand{\coqdocvar}[1]{{\color{\coqdocvarcolor}{\texttt{#1}}}}
\renewcommand{\coqdoccst}[1]{{\color{\coqdoccstcolor}{\texttt{#1}}}}

\renewcommand{\coqdocmod}[1]{{\color{\coqdocmodcolor}{\textsl{\texttt{#1}}}}}

\renewcommand{\coqdocax}[1]{{\color{\coqdocaxcolor}{\textsl{\texttt{#1}}}}}

\renewcommand{\coqdocind}[1]{{\color{\coqdocindcolor}{\textbf{\texttt{#1}}}}}

\renewcommand{\coqdocconstr}[1]{{\color{\coqdocconstrcolor}{\texttt{#1}}}}

\renewcommand{\coqdoctac}[1]{{\color{\coqdoctaccolor}{\texttt{#1}}}}

\makeatletter
\renewcommand{\ALG@beginalgorithmic}{\small}
\makeatother

\spnewtheorem{ourclaim}{Claim}{\bfseries}{\itshape}
\spnewtheorem{observation}{Observation}{\bfseries}{\itshape}
\spnewtheorem*{nonumhypo}{Hypothesis}{\bfseries}{\itshape}
\spnewtheorem*{nonumremark}{Remark}{\bfseries}{\itshape}

\bibliography{ref}{}

\ifdefined\doubleblind
\bibliography{ref-blind}{}
\else
\bibliography{ref-unblind}{}
\fi

\begin{document}
\title{Improved Tristate Multiplication With Formalization in Rocq}

\ifdefined\doubleblind
\author{Author(s) hidden for double-blind review}
\else
\author{Nandakumar Edamana, Piyush P Kurur, Unnikrishnan Cheramangalath}
\institute{Indian Institute of Technology Palakkad, India\\
  \email{112314001@smail.iitpkd.ac.in},\\
  \email{\{ppk, unnikrishnan\}@iitpkd.ac.in}}
\fi

\maketitle

\begin{abstract}
  We give a new multiplication algorithm for tristate numbers
  improving upon the state-of-the-art implementation from the Linux
  kernel in terms of precision. Our algorithm is significantly more
  precise as shown by experimental evaluation (at the peak, giving
  better results in $95.59\%$ cases compared to the previous work for
  31-bit samples). Importantly, we achieve this additional precision
  with performance comparable to the previous algorithm, as
  demonstrated by benchmarks. Our algorithm is now part of the
  upstream Linux kernel.

  We also provide the first machine-checked proofs for tnum arithmetic
  operations that work for all bit widths. Our proofs include the
  soundness of the new multiplication algorithm and the soundness and
  optimality of supporting operations.

  \keywords{Abstract interpretation \and Formal verification \and
    Kernel extensions \and eBPF \and Rocq}
\end{abstract}

\nke{Git: \input{gitinfo}}
\nke{TACAS 2027 Deadline: Oct 15, 2026 (AoE) = Oct 15, 2026 17:30 IST}
\nke{TACAS 2027 Page limit: 18 pages}

\iftrue
\else
\section*{TODO Address SAS26 Review}
\begin{enumerate}
\item
  The tnum domain forms a Galois insertion ($\alpha \circ \gamma$ = id, as the authors note), yet this standard terminology is never used.

\item Other concerns regarding the Background section

\item inlcude links to the repo (first, make sure to specify the
  license in the repo)

\item Compare our tnum\_union with prior literature suggested by the reviewers

\item Artifact: cycle count estimation

\item
  The wrap-around results in the loss of higher bits in the
  implementations of both the new and old algorithms. But for a randomly
  picked 64-bit number, the odds of it having non-zero bits at higher
  positions (MSB-side) is very high. This is why both algorithms show
  the same relative precision at 64 bits.

\end{enumerate}

\section*{Other TODO}
\begin{enumerate}
\item
  IMPORTANT: when using ebpf-samples from prevail: .elf files are also available now (not just .o files)
\end{enumerate}

\section*{Present in the Rocq Code, Missing Here}
\begin{enumerate}

\item \lstinline{tnum_union_optimal_by_subset}

\item \lstinline{tnum_sub_sound}, modulo proof for the full subtractor
equation

\end{enumerate}

\section*{Changes Made After SAS26 Review (Until End of Aug 2026)}
\begin{itemize}
\item Mention union-related proofs, updated addition-related
  proofs, etc. in the abstract and the intro (contribs).

\item Drop the \lstinline{tnum_add} optimality proof from Harishankar
  et al. and use our version.

\item For soundness, simplify some lemmas by defining
  \lstinline{sound2}, a common soundness predicate for 2-ary tnum
  operations.

\item Add some explanation regarding the proofs (not just the
  lemmas).

\item Start with the mul-related proofs and move the addition-related
  ones to the last.

\item Address SAS26 review concern: \lstinline{nat} is not defined.

\item Add remark: references like \lstinline{tnum_mul} could mean the
  function, \rocq{} code, or the algorithm.

\item Address CAV26/SAS26 review concern: make it clear that the code
  is/will be publicly available.

\item Add stats: mul usage in real-world eBPF programs.

\item Address CAV26/SAS26 review concern: theoretical complexity is
  not mentioned.

\item Fix: missing note on Rocq encoding being manual.

\item Abstract: \emph{correctness} $\to$ \emph{soundness}

\item Move the importance of eBPF back into the intro from Related
  Works (similar to the first draft).

\item Fix: duplicate/alt definition of tnum union and its soundness.

\item Add: how \lstinline{tnum_union} works.

\item Partially rewrite the section Tnum Arithmetic in the Linux
  Kernel. The main change is to move from an array-of-structs view to
  a struct-of-arrays view, because that is how both Linux and our
  proofs do it.
\end{itemize}
\fi 

\setcounter{tocdepth}{10}
\listoftodos{}

\section{Introduction}

A \emph{trit} is a bit that is either 0, 1, or $\mu$ (denoting
uncertainty). \emph{Tnums} of length $n$ are sequences of $n$-trits
and represent $n$-bit numbers some of whose bits are
uncertain. Functions in general and arithmetic operations in
particular on $n$-bit numbers have their corresponding variants for
tnums which take into consideration these uncertain bit positions.

In the framework of \emph{abstract
interpretation}~\cite{CousotCousot76-1}, tnums together with the
associated arithmetic operations form an \emph{abstract domain} and
are used to perform various semantic analyses like bounds checking on
programs.

Tnums are extensively used in the Linux kernel \ebpf{}~\cite{ebpf}
verifier which checks a user-given \ebpf{} program before running it
directly in the kernel space. The extensive capabilities for tracing
and instrumentation that \ebpf{} provides make it attractive for
hyperscalers looking for improved visibility~\cite{ebpfCaseStudy1,
  ebpfCaseStudy2}, performance~\cite{ebpfCaseStudy3}, and
efficiency~\cite{ebpfCaseStudy4, ebpfCaseStudy5}. Sound and fast
implementation of tnum arithmetic is therefore crucial for the safety
of such systems.

This work presents a new algorithm and proofs for tnum arithmetic and
is thus directly relevant for the \ebpf{} subsystem. Our contributions
are:

\begin{description}
\item[New multiplication algorithm:] A novel algorithm for tnum
  multiplication that is simpler in description, equally efficient in
  running time, and importantly, more precise than the previously
  known state-of-the-art algorithm~\cite{hariTnum} as shown by
  experimental evaluation.

  \begin{itemize}
  \item At 8-bit width, our algorithm is \emph{optimal} for 76\% of
    the cases, whereas the previous algorithm was optimal for 22\% of
    the cases only. Our algorithm is also \emph{more precise} than the
    previous version in 69.8\% of the cases.
  \item For higher bit widths, we could only perform random
    sampling. At 31-bit width, our algorithm is more precise in 95.6\%
    of the cases as compared to the previous version when tested
    against 10 million random tnum pairs. Comparison with the optimal
    multiplication is not feasible as the brute force evaluation is
    prohibitively slow.
  \item Even in incomparable cases, our algorithm has better precision
    on higher bits.
  \end{itemize}

  Real-world eBPF programs involve significant use of 32-bit
  multiplication (see Section~\ref{sec:mul-usage-stats}). The
  improvement in precision enables the Linux kernel \ebpf{} verifier
  to accept more number of valid programs, as demonstrated in
  Appendix~\ref{apx:test-program}.

\item[Contribution to the Linux Kernel:] Our C implementation of the
  new tnum multiplication algorithm is now part of the upstream Linux
  kernel, starting from Linux 6.18.

\item[Formal proofs:] In terms of formal guarantees, we provide in
  \rocq{}~\cite{rocq}:

  \begin{itemize}
  \item The \emph{soundness} (see Definition~\ref{def:sound}) proof of
    the new multiplication algorithm and its dependencies.
  \item A simpler and straightforward optimality (see
    Definition~\ref{def:optimal}) proof for tnum addition.
  \item Soundness and optimality proofs for the tnum union algorithm
    newly added to the Linux kernel.
  \item Proof of correctness of a modified C implementation of tnum
    multiplication and all the supporting operations that employ the
    same algorithms used in the Linux kernel.
  \end{itemize}

  Previous attempts at machine-checked formal proofs were based on
  SAT/SMT solvers. In the case of soundness of multiplication, the SMT
  solver was unable to scale beyond 8 bits, whereas the kernel version
  uses 64-bit integers~\cite{hariTnum}. Our proofs, on the other hand,
  work for all lengths and thus fill an important gap in the formal
  proofs of tnum arithmetic.
\end{description}

\section{Background}
\label{sec:background}

We use $\Bit = \{ 0 , 1 \}$ to denote the set of bits.  An $n$-bit
vector is an element of the set $\Bvec{n}$. As convention, we use
bold-faced lower case letters like $\mathbf{x}$ for $n$-bit vectors
and $\mathbf{x}_i$ for its $i$-th bit, for every position $i$.

A bit vector $\mathbf{x}$ can also be seen as the representation of
the natural number $\natural \mathbf{x} = \sum_i 2^i \mathbf{x}_i$ in
binary. Conversely, for a fixed bit size $n$, a natural number $N$ can
be associated with the $n$-bit vector whose $i$-th bit is defined as:
\[ \bit{i}{N} = \lfloor N/2^i \rfloor \mod 2. \]

The $n$-bit vectors associated with two natural numbers $N$ and $M$
are the same if they lie in the same residue class modulo
$2^n$. Therefore, we often identify the set of natural numbers less
than $2^n$ with the corresponding set of $n$-bit vector.

A \emph{tristate bit}, or \emph{trit} for short, is a bit which might
be uncertain. Mathematically, a trit is an element of the set $
\mathbb{T} = \{ \trtO, \trtI, \mu\}$ where $\mu$ denotes
uncertainty. Analogous to $n$-bit vectors, the $n$-length
\emph{tristate numbers}, or \emph{tnums} for short, is the set
$\Tnum{n}$. As convention, we use bold-faced upper case letters like
$\mathbf{P}$ for $n$-length tnums and use $\mathbf{P}_i$ to denote the
associated trit at position $i$. Thus a tnum $\mathbf{P}$ can be seen
as a bit vector with uncertainties at positions $i$ where
$\mathbf{P}_i = \mu$.

A trit can also be thought of as a \emph{pattern} to \emph{match} a
single bit where $\trtO$ and $\trtI$ matches $0$ and $1$ respectively
and $\mu$ can match either of the two. Alternatively, a trit can be
associated with the nonempty subset of bits that it \emph{matches}. In
other words, $\trtO$, $\trtI$, and $\mu$ can be identified with the
non-empty subsets $\{0\}$, $\{1\}$, and $\{0,1\}$ respectively. Thus,
the set $\Trit$ is the set $2^\Bit \setminus \{ \emptyset \}$ of
nonempty subsets of $\Bit$.

The matching relation can be extended to $n$-length tnums as follows.
A tnum $\mathbf{P} \in \Tnum{n}$ matches $\mathbf{x} \in \Bvec{n}$ if
$\mathbf{x}_i \in \mathbf{P}_i$ for every position $i$. Due to the
identification of bitvectors and natural numbers above, a tnum
$\mathbf{P} \in \Tnum{n}$ matches the the natural number $N < 2^n$ if
it each of its bits $N_i \in \mathbf{P}_i$.

\subsection{The Abstract Domain of Tnums}
\label{sec:tnum-absdom}

The matching relation described above allows us to think of elements
of the set $\Tnum{n}$ as subsets of natural numbers. Within the
framework of abstract interpretation, the set $\Tnum{n}$ is an
abstract domain with the set $\Cdom = 2^\mathbb{N} \setminus \{
\emptyset \}$ of all nonempty subsets of natural numbers serving as
the associated \emph{concrete domain}. The association between these
two domains is formalized using the \emph{abstraction} and
\emph{concretization} functions defined below.

\begin{description}
\item[Abstraction function:] For every position $i$, define the
  function $\alpha_i : \Cdom \to 2^\Bit$ as follows:
  \[ \alpha_i(S) = \{ \bit{i}{N} \ | \  N \in S \}. \]

  Notice that every subset $S \in \Cdom$ is non-empty and hence
  $\alpha_i(S)$ is either the set $\{0\}$, $\{ 1 \}$ or $\{ 0 , 1\}$
  which we identify with the trits $\trtO$, $\trtI$, and $\mu$
  respectively. Thus $\alpha_i$ can be seen as a function from $\Cdom$
  to $\Trit$. For a subset $S \in \Cdom$ we define $\alpha (S)$ as
  that tnum whose $i$-th trit is $\alpha_i(S)$. In other words, the
  abstraction function $\alpha : \Cdom \to \Tnum{n}$ is defined as
  \[ \alpha(S) = (\alpha_0(S),\ldots, \alpha_{n-1}(S)). \]

\item[Concretization function:] The function $\gamma: \Tnum{n} \to
  \Cdom$ maps a given tnum $\mathbf{P} \in \Tnum{n}$ to the set of all
  natural numbers corresponding to bit vectors matching
  $\mathbf{P}$. More precisely,
  \[ \gamma (\mathbf{P}) = \{ N \in \mathbb{N} \ | \ \mathbf{P} \ \textrm{matches} \ N \}. \]

\end{description}

The concept of abstraction and concretization is illustrated in
Figure~\ref{fig:eg-tnum-alpha-gamma} with an example. While
$\alpha(\gamma(\mathbf{P}))$ is always $\mathbf{P}$, notice from the
example that $\gamma(\alpha(C))$ can be a superset of $C$.

\begin{figure}
  \begin{center}
    \begin{tabular}{c|c|c}
      Concrete set & Abstraction & Concretization \\
      $C = \{ 0011, 0101, 0111 \}$ & \
      $\mathbf{P} = \alpha(C) = 0\mu\mu1$ & \
      $\gamma(\mathbf{P}) = \{ 0001, 0011, 0101, 0111 \}$
    \end{tabular}
  \end{center}

  \caption{An example illustrating abstraction and concretization}
  \label{fig:eg-tnum-alpha-gamma}
\end{figure}

The set theoretic flavor of the abstract domain $\Tnum{n}$ formalized
through the concretization function $\gamma$ means that the subset
relation on the concrete domain $\Cdom$ has a natural counter part on
the abstract domain.

\begin{definition}[Inclusion order]
  The inclusion order $\tle$ on $\Tnum{n}$ is defined as follows. For
  tnums $\mathbf{P}$ and $\mathbf{Q}$, $\mathbf{P} \tle \mathbf{Q}$
  \emph{if and only if} for every position $i$, the
  corresponding trits as subsets satisfy the inclusion $\mathbf{P}_i \subseteq
  \mathbf{Q}_i$.
\end{definition}

It is easy to see that for all tnums $\mathbf{P}$ and $\mathbf{Q}$, we
have $\mathbf{P} \tle \mathbf{Q}$ if and only if $\gamma(\mathbf{P})
\subseteq \gamma (\mathbf{Q})$.

In the context of abstract interpretation, we replace the concrete
operator $\star$ on the concrete domain $\Cdom$ by a \emph{sound
approximation} $\star'$ on the abstract domain $\Tnum{n}$.

\begin{definition}[Sound approximation]\label{def:sound}
  Let $\star$ be any operator on $\Cdom$. We say an operator $\star'$
  on $\Tnum{n}$ is a \emph{sound approximation} if for all tnums
  $\mathbf{P}$ and $\mathbf{Q}$ in $\Tnum{n}$, we have

  \[ \gamma(\mathbf{P}) \star \gamma(\mathbf{Q})  \subseteq \gamma\left( \mathbf{P} \star' \mathbf{Q} \right). \]
\end{definition}

A tnum is used to keep track of uncertain bits in a bit vector. Our
definition of soundness captures the intuition that using a larger set
(i.e., larger uncertainty) for analysis will not compromise correctness
though it could lose out on precision.  Soundness, while necessary for
the correctness, is not always sufficient to get meaningful
results. For example, the trivial operator that maps all pairs of
tnums to a sufficiently long tnum with $\mu$ in all positions is a sound
approximation of any operator $\star$ though a useless one. We would
like to use the most accurate approximation whenever possible. To
formulate this notion, we first define the inclusion order on
operators as follows.

\begin{definition}[Inclusion order for operators]
  Let $\star_1$ and $\star_2$ be operators on $\Tnum{n}$, we say
  $\star_1 \tle \star_2$ if for all tnums $\mathbf{P}$ and
  $\mathbf{Q}$, $\mathbf{P} \star_1 \mathbf{Q} \tle \mathbf{P} \star_2
  \mathbf{Q}$.
\end{definition}

Since inclusion order on tnums translate to set theoretic inclusion of
the concrete domain, the ordering $\star_1 \tle \star_2$ holds if and
only if for all $\mathbf{P}$ and $\mathbf{Q}$, $\gamma(\mathbf{P}
\star_1 \mathbf{Q}) \subseteq \gamma(\mathbf{P} \star_2
\mathbf{Q})$. Based on the above inclusion ordering, we define the
optimal approximation as follows.

\begin{definition}[Optimal approximation]\label{def:optimal}
  Let $\star$ be an operator on $\Cdom$.  A \emph{sound} approximation
  $\overline{\star}$ of $\star$ on $\Tnum{n}$ is \emph{optimal} if for
  every sound approximation $\star'$ of $\star$, we have $\overline{\star}
  \tle \star'$.
\end{definition}

Optimal approximation always exists.~\todo{no need to mention Galois
  connection?} In fact, for any operator $\star$ on the concrete
domain $\Cdom$, it is easy to show that the operator
$\overline{\star}$ defined by the equation below is optimal.

\[ \mathbf{P} \mathbin{\overline{\star}} \mathbf{Q} = \alpha \left(\gamma(\mathbf{P}) \star \gamma(\mathbf{Q}) \right). \]

Consider the operator $\star$ on $\mathbb{N}$. This operator can be
\emph{lifted} to the concrete domain $\Cdom$ (also denoted by $\star$)
defined as follows:
\[ S \star T = \{ N \star M \ | \ N \in S, M \in T \}. \]

Therefore, for such an operator, by its approximation on $\Tnum{n}$ we
mean the approximation of this lifted operator.

\subsection{Tnum Arithmetic in the Linux Kernel}~\label{sec:tnum-in-linux}

\newcommand{\tvec}[1]{\mathbf{#1}}

\newcommand{\tveci}[1]{\mathbf{#1}_i}

Inside the Linux kernel, an $n$-length tnum $\tvec{P}$ is represented
as a \emph{value-mask pair} $P=(\tv{P}, \tm{P})$ consisting of two
$n$-length bit vectors in the following way:

\[
\begin{aligned}
\tveci{P} = \trtO \iff \ithv{P} = 0 \land \ithm{P} = 0 \\
\tveci{P} = \trtI \iff \ithv{P} = 1 \land \ithm{P} = 0 \\
\tveci{P} = \mu \iff \ithv{P} = 0 \land \ithm{P} = 1
\end{aligned}
\]

That is, the $i$-th mask bit $\ithm{P}$ being $0$ signifies a known
bit (value given by $\ithv{P}$) and being $1$ signifies an unknown bit
[for example, the tnum $0\mu\mu1$ is represented as $(0001, 0110)$]. We
say that a value-mask representation is \emph{well-formed} if and only
if $\ithv{P} = 0$ whenever $\ithm{P} = 1$ (i.e., the value bit should
be $0$ whenever the corresponding trit is $\mu$) and \emph{ill-formed}
otherwise.

All the tnum operations in the Linux kernel is part of a single C file
library~\cite{linux-tnum-h, linux-tnum-c} with tnums represented as a
value-mask pair using a \lstinline{struct} called \lstinline{tnum}
having two \lstinline{u64} (64-bit unsigned integer) fields
\lstinline{value} and \lstinline{mask}. The functions exposed
by this library assume the well-formedness of the input arguments and
ensure the well-formedness of the results by design.

Our focus is on the arithmetic function \lstinline{tnum_mul} and
related operations exposed by the Linux tnum library.  While addition
and multiplication are straightforward to implement when tristate
numbers are represented as trit vectors, they are difficult to
implement for the Linux \lstinline{tnum} type, using word-level
operations (a requirement for performance reasons). This can bee seen
from the the Linux kernel implementation of tnum addition, given in
Listing~\ref{lst:tnum-add}.

\begin{lstlisting}[caption={Linux function defining tnum addition},label=lst:tnum-add,
    language=C]
struct tnum tnum_add(struct tnum a, struct tnum b)
{
	u64 sm, sv, sigma, chi, mu;

	sm = a.mask + b.mask;
	sv = a.value + b.value;
	sigma = sm + sv;
	chi = sigma ^ sv;
	mu = chi | a.mask | b.mask;
	return TNUM(sv & ~mu, mu);
}
\end{lstlisting}

\section{Our Algorithm for Tnum Multiplication}\label{sec:new-multiplication}

Our tnum multiplication algorithm given in
Algorithm~\ref{lst:algo-tnum-mul-new} builds upon the shift-and-add
multiplication algorithm for bit vectors. Given two tnums $\mathbf{A}$
and $\mathbf{B}$, we have an accumulator for holding the partial
product, initialized to $\trtO$. For each bit position $i$, if the
trit $B_i$ is certain (i.e., is either $\trtO$ or $\trtI$), we proceed
like in the case of bit vectors. We add (using tnum addition)
$\Lshift{i}{\mathbf{A}}$ to the accumulator if the trit $B_i$ is
$\trtI$ and leave it intact if it is $\trtO$. However, if the trit
$B_i$ is $\mu$, then a naive implementation can lead to massive loss
in precision, as observed in the quote below.

\begin{quote}
  How can we incorporate correlation in unknown bits across partial
  products? For example, multiplying P = 11, Q = ${\mu}1$ produces the
  partial products T1 = 11, T2 = ${\mu}{\mu}0$. However, the two $\mu$
  trits in T2 are concretely either both 0 or both 1, resulting from
  the same $\mu$ trit in Q. Failing to consider this in the addition
  makes the result imprecise.

  \hfil --- \citeauthor{hariTnum}~\cite{hariTnum}
\end{quote}

The main idea behind our algorithm is to handle the correlation
mentioned above. For each trit $B_i$ that is $\mu$, we compute two
possible partial products $acc_0$ and $acc_1$ corresponding to the
cases when the $i$-th bit is $0$ and $1$ respectively. The actual
accumulator is the union of $acc_0$ and $acc_1$
(Lines~\ref{lstln:start-correl}--\ref{lstln:end-correl} in
Algorithm~\ref{lst:algo-tnum-mul-new}). We repeat the process of
shift-and-add until both the value and the mask components of the
multiplier becomes 0 (see the \lstinline{while} loop in
Algorithm~\ref{lst:algo-tnum-mul-new}).

\begin{algorithm}
  \caption{Our tnum multiplication algorithm}
  \label{lst:algo-tnum-mul-new}
  \begin{algorithmic}[1]
  \Procedure{tnum\_mul}{$a, b$}
	\State $acc \gets TNUM(0, 0)$

	\While{$a.value$ \textbf{or} $a.mask$}
		\If{$LSB(a.value) == 1$}
			\State $acc \gets tnum\_add(acc, b)$
		\ElsIf{$LSB(a.mask) == 0$}
			\State $acc \gets acc$ \Comment{No change in acc}
	        \ElsIf{$LSB(a.mask) == 1$} \label{lstln:start-correl}
			\State $acc_0 \gets acc$ \Comment{Possibility: concrete $a_i$ is 0}
			\State $acc_1 \gets tnum\_add(acc, b)$ \Comment{Possibility: concrete $a_i$ is 1}
			\State $acc \gets tnum\_union(acc_0, acc_1)$ \Comment{Final partial product}
                \EndIf \label {lstln:end-correl}

		\State $a \gets tnum\_rshift(a, 1)$
		\State $b \gets tnum\_lshift(b, 1)$
       \EndWhile

       \State \textbf{return} acc
  \EndProcedure

  \Procedure{tnum\_union}{$a, b$}
	\State $v \gets a.value$ \& $b.value$
	\State $m \gets (a.value \wedge b.value) \mathbin{|} a.mask \mathbin{|} b.mask$

	\State \textbf{return} $TNUM(v $ \& $\sim m, m)$
  \EndProcedure
  \end{algorithmic}
\end{algorithm}

The listing of our algorithm contains some redundant operations like
$acc \leftarrow acc$ to make the intuition clear. Removing such
operations, we implement the new algorithm as shown in
Listing~\ref{lst:tnum-mul-new}.

\begin{figure}
\begin{lstlisting}[caption={Our code implementing the new algorithm for tnum multiplication},label=lst:tnum-mul-new,
    language=C]
struct tnum tnum_mul(struct tnum a, struct tnum b)
{
	struct tnum acc = TNUM(0, 0);

	while (a.value || a.mask) {
		if (a.value & 1)
			acc = tnum_add(acc, b);
		else if (a.mask & 1)
			acc = tnum_union(acc, tnum_add(acc, b));

		a = tnum_rshift(a, 1);
		b = tnum_lshift(b, 1);
	}
	return acc;
}
\end{lstlisting}
\end{figure}

\subsubsection{Union Operation.}
The union of two tnums is an abstraction of the unions of the concrete
sets they represent. Mathematically,

\begin{definition}[Union]
  \label{def:tnum-union}
  For tnums $\mathbf{P}$ and $\mathbf{Q}$,
\[
\mathbf{P} \cup \mathbf{Q} = \alpha(\gamma(\mathbf{P}) \cup \gamma(\mathbf{Q})).
\]
\end{definition}

An efficient algorithm for this operation is given in
Aglorithm~\ref{lst:algo-tnum-mul-new}, which achieves the above
without actually computing the concrete sets represented by the input
tnums. The intuition behind the algorithm is discussed in
Section~\ref{sec:tnum-union-sound} along with formal proofs of
soundness and optimality.

\subsubsection{Time Complexity.}
Our \lstinline{tnum_mul} uses $O(n)$ word operations (where $n$ is the
bit width), similar to the previous work~\cite{hariTnum}. The addition
of \lstinline{tnum_union} in the new algorithm does not change the
overall complexity since it involves $O(1)$ word operations.

Technically, for arbitrary precision operands, each word operation
itself is of $O(n)$. Thus in terms of Turing Machine complexity, both
algorithms should be considered to be of $O(n^{2})$ complexity.

\subsubsection{Example.}
Figure~\ref{fig:new-mul-example} illustrates with an example how our
algorithm deals with any correlation in unknown bits. Here, the
multiplicand is $11$, and when a $\mu$ trit is encountered in the
multiplier, the previous algorithm contributes to the accumulator
$\mu\mu$ (with appropriate shift), while the new algorithm considers
both the possibilities (i.e., two possible partial products after
adding $00$ and $11$ respectively, with appropriate shift), resulting
in better precision.

\newcommand{\bitcolpad}{0em}

\newcommand{\bitselected}[1]{\textcolor{blue}{\textbf{#1}}}

\begin{figure}
  \caption{Illustration of the new algorithm handling correlation in
    unknown bits}
  \label{fig:new-mul-example}
  \begin{tcolorbox}[colback=white]
    Previous algorithm:

    \begin{center}
      \begin{tabular}{ p{1em} p{\bitcolpad} p{\bitcolpad} p{\bitcolpad} p{\bitcolpad} p{\bitcolpad} }
        \textbf{b} & & & 1 & 1 & $\times$ \\
        \textbf{a} & & & $\mu$ & 1 \\ \hline
        \textbf{T1} & & & 1 & 1 & \\
        \textbf{T2} & & $\mu$ & $\mu$ & \\ \hline
        & $\mu$ & $\mu$ & $\mu$ & 1
      \end{tabular}
    \end{center}

    \rule{\textwidth}{1pt}

    New algorithm:

    \begin{center}
      \begin{tabular}{ p{0.18\textwidth} p{1em} p{0.18\textwidth} p{1em} p{0.3\textwidth} }
        $acc_0$ & & $acc_1$ & & {Final $acc$} \\ \hline \\
        \begin{tabular}{ p{\bitcolpad} p{\bitcolpad} p{\bitcolpad} p{\bitcolpad} p{\bitcolpad} p{\bitcolpad} }
          & & 1 & 1 & $\times$ \\
          & & \bitselected{0} & 1 \\ \hline
          & & 1 & 1 & \\
          & 0 & 0 & \\ \hline
          0 & 0 & 1 & 1
        \end{tabular} &
        $\to$ &
        \begin{tabular}{ p{\bitcolpad} p{\bitcolpad} p{\bitcolpad} p{\bitcolpad} p{\bitcolpad} p{\bitcolpad} }
          & & 1 & 1 & $\times$ \\
          & & \bitselected{1} & 1 \\ \hline
          & & 1 & 1 & \\
          & 1 & 1 & \\ \hline
          1 & 0 & 0 & 1
        \end{tabular} &
        $\to$ &
        $\text{tnum\_union}(0011, 1001) = {\mu}0{\mu}1$
      \end{tabular}
    \end{center}

  \end{tcolorbox}
\end{figure}

\section{Formalization in Rocq}

We now give a detailed overview of the \rocq{} formalization and
subsequent verification of different properties of the tnum operations
relevant to this work as follows. This includes:

\begin{itemize}
\item The soundness of our new multiplication algorithm.
\item The soundness of the existing tnum addition and the newly
  contributed union operation. These soundness proofs are needed for
  the soundness of the new multiplication algorithm.
\item The optimality of the existing addition operation and the new
  union operation. Although these optimality proofs are not relevant
  for the multiplication algorithm (which is not optimal anyway), as
  an independent property of the addition and union operations, they
  are interesting in their own right.~\todo{But the optimality of the
    supporting operations contribute to the precision of
    multiplication.}
\end{itemize}

The complete source of this formalization is available at our
public{\ifdefined\doubleblind \footnote{For the double-blind review,
  we replace this URL with an anonymized Docker image containing all
  the relevant source code and development tools.} \fi}
repository~\cite{trirocq-repo}.

\subsubsection{End-to-end Verification.}
Our soundness and optimality proofs only address the \rocq{} versions
of the algorithms. To lower the formal guarantees to the C level, we
provide an additional C implementation amenable to verification using
the Verified Software Toolchain~\cite{vst} (which in turn depends on
the CompCert~\cite{compcert} C semantics). Following the standard
procedure, we convert this C code to a \rocq{} representation using
\lstinline{clightgen} (part of CompCert) and show that the C functions
are semantically equivalent to the tnum operations defined in \rocq{}.
We have proofs for tnum addition, union, shifts, and multiplication.

The key difference in the modified C version is the use of pointers
for passing tnums. The Linux kernel tnum routines pass tnums by
copying the corresponding structs directly, a feature not supported by
VST. This rules out the direct verification of the kernel version of
the above algorithms. However, if end-to-end verification is desired,
our version can be used in the kernel with a thin wrapper. Besides,
independent projects that need a verified C version of tnum arithmetic
can directly use our alternative version.

\subsection{Modeling Machine Integers}
\label{sec:model-mach-int}

Inside the Linux kernel, the tnum arithmetic is defined using 64-bit
unsigned integers. Our formalization is in terms of the generalized
type \lstinline{uint SIZE} of unsigned integers parameterized on its
bit-size \lstinline{SIZE} with the arithmetic operations on the
underlying type \lstinline{Z} modulo $2^{\mathtt{SIZE}}$ serving as
the arithmetic operations on \lstinline{uint SIZE}.

\rocqlisting{
  \coqdocnoindent
    \coqdockw{Section} \coqdocvar{uint}.\coqdoceol
    \coqdocindent{1.00em}
    \coqdockw{Variable} \coqdocvar{SIZE} : \coqdocvar{nat}.\coqdoceol
    \coqdocindent{1.00em}
    \coqdockw{Definition} \coqdocvar{modulus} := 2 \^{} \coqdocvar{Z.of\_nat} \coqdocvar{SIZE}.\coqdoceol
    \coqdocemptyline
    \coqdocindent{1.00em}
    \coqdockw{Definition} \coqdocvar{uint} := \{ \coqdocvar{x} \ensuremath{|} 0 \ensuremath{\le} \coqdocvar{x} \texttt{<} \coqdocvar{modulus} \}.\coqdoceol
    \coqdocemptyline
    ...\coqdoceol

    \coqdocnoindent
    \coqdockw{End} \coqdocvar{uint}.\coqdoceol
  }

Directly using the standard definitions of \lstinline{uint} arithmetic
mentioned above makes the inductive proofs difficult. Instead, we give
alternative implementations of these arithmetic operations by viewing
\lstinline{uint SIZE} as a bit vector. These bit-vector-style
implementations make use of the helper functions \lstinline{uint_ith}
that returns the $i$-th bit of a \lstinline{uint}. Finally, to connect
these soundness proofs to the standard modular-arithmetic-based
definitions, we prove their equivalence to the corresponding
bit-vector-style implementation.

\subsection{Modeling Tnum}
\label{sec:model-tnum}

Recall from Section~\ref{sec:tnum-in-linux} the value-mask
representation of tnums used in the Linux kernel and the notion of
well-formedness associated with it. We model these in \rocq{} as given
below. The helpers \lstinline{tnum.m} and \lstinline{tnum.v} extract
the mask word and the value word from a \lstinline{tnum} object
respectively.

\rocqlisting{
  \coqdocnoindent
  \coqdockw{Module} \coqdocvar{tnum}.\coqdoceol
  \coqdocindent{1.00em}
  \coqdockw{Variant} \coqdocvar{t} \coqdocvar{SIZE} := \coqdocvar{cons} (\coqdocvar{v} : \coqdocvar{uint} \coqdocvar{SIZE}) (\coqdocvar{m} : \coqdocvar{uint} \coqdocvar{SIZE}).\coqdoceol
  ...\coqdoceol
  \coqdocnoindent

  \coqdocindent{1.00em}
  \coqdockw{Definition} \coqdocvar{wellformed} \{\coqdocvar{SIZE}\} (\coqdocvar{tn} : \coqdocvar{t} \coqdocvar{SIZE}) :=\coqdoceol
  \coqdocindent{2.00em}
  \coqdockw{\ensuremath{\forall}} \coqdocvar{i} (\coqdocvar{hidx} : \coqdocvar{i} \texttt{<} \coqdocvar{SIZE}),\coqdoceol
  \coqdocindent{3.00em}
  \coqdocvar{uint\_ith} (\coqdocvar{m} \coqdocvar{tn}) \coqdocvar{hidx} = \coqdocvar{one} \ensuremath{\rightarrow} \coqdocvar{uint\_ith} (\coqdocvar{v} \coqdocvar{tn}) \coqdocvar{hidx} = \coqdocvar{zero}.\coqdoceol
  ...\coqdoceol

  \coqdocnoindent
  \coqdockw{End} \coqdocvar{tnum}.\coqdoceol
}

\begin{nonumremark}
  Our proofs of soundness and optimality of tnum operations assert the
  well-formedness of the operands as a precondition. This is omitted
  in explanations for brevity.
\end{nonumremark}

Now, we define the tnum membership as follows, which is needed in the
lemmas we prove.

\rocqlisting{\coqdockw{Definition} \coqdocvar{ingamma} \{\coqdocvar{SIZE}\} (\coqdocvar{x} : \coqdocvar{uint} \coqdocvar{SIZE}) (\coqdocvar{T} : \coqdocvar{tnum.t} \coqdocvar{SIZE}) :=\coqdoceol
\coqdocindent{1.00em}
\coqdockw{\ensuremath{\forall}} \coqdocvar{i}, \coqdocvar{i} \texttt{<} \coqdocvar{SIZE} \ensuremath{\rightarrow}\coqdoceol
\coqdocindent{6.00em}
\coqdocvar{tnum.ith\_m} \coqdocvar{T} \coqdocvar{i} = \coqdocvar{false} \ensuremath{\rightarrow} \coqdocvar{uint\_ith} \coqdocvar{x} \coqdocvar{i} = \coqdocvar{tnum.ith\_v} \coqdocvar{T} \coqdocvar{i}.\coqdoceol
}

The \lstinline{ingamma} predicate states the property of every bit
vector $\mathbf{x} \in \gamma(\mathbf{T})$ that for every position
$i$, the $i$-th mask bit of $\mathbf{T}$ is \lstinline{zero} implies
$x_i$ and the $i$-th value bit of $\mathbf{T}$ are the same.

\subsection{Modeling Soundness}~\todo{include optimality?}

For our proofs, we use an alternative definition of soundness that
follows from Definition~\ref{def:sound}. This version is more
convenient in this context.

\begin{definition}[Alternative definition of soundness]
  \label{def:soundness-for-proofs}
  Given a function $f$ operating on bit vectors and a function $F$
  operating on tnums, $F$ is a sound approximation of $f$ if
  $f(\mathbf{p}, \mathbf{q}) \in \gamma(F(\mathbf{P}, \mathbf{Q}))$
  for all $\mathbf{p} \in \gamma(\mathbf{P})$ and $\mathbf{q} \in
  \gamma(\mathbf{Q})$.
\end{definition}

In \rocq, we define \lstinline{sound2}, a specialization of the above
for 2-ary functions as follows:

\rocqlisting{\coqdockw{Definition} \coqdocvar{sound2} \coqdocvar{SIZE}\coqdoceol
\coqdocindent{1.00em}
(\coqdocvar{f} : \coqdocvar{uint} \coqdocvar{SIZE} \ensuremath{\rightarrow} \coqdocvar{uint} \coqdocvar{SIZE} \ensuremath{\rightarrow} \coqdocvar{uint} \coqdocvar{SIZE})\coqdoceol
\coqdocindent{1.00em}
(\coqdocvar{F} : \coqdocvar{tnum.t} \coqdocvar{SIZE} \ensuremath{\rightarrow} \coqdocvar{tnum.t} \coqdocvar{SIZE} \ensuremath{\rightarrow} \coqdocvar{tnum.t} \coqdocvar{SIZE}) :=\coqdoceol
\coqdocindent{1.00em}
\coqdockw{\ensuremath{\forall}} (\coqdocvar{p} \coqdocvar{q} : \coqdocvar{uint} \coqdocvar{SIZE}) \coqdocvar{P} \coqdocvar{Q},\coqdoceol
\coqdocindent{2.00em}
\coqdocvar{tnum.wellformed} \coqdocvar{P} \ensuremath{\rightarrow} \coqdocvar{tnum.wellformed} \coqdocvar{Q} \ensuremath{\rightarrow}\coqdoceol
\coqdocindent{2.00em}
\coqdocvar{ingamma} \coqdocvar{p} \coqdocvar{P} \ensuremath{\rightarrow} \coqdocvar{ingamma} \coqdocvar{q} \coqdocvar{Q} \ensuremath{\rightarrow}\coqdoceol
\coqdocindent{2.00em}
\coqdocvar{tnum.wellformed} (\coqdocvar{F} \coqdocvar{P} \coqdocvar{Q}) \ensuremath{\land} \coqdocvar{ingamma} (\coqdocvar{f} \coqdocvar{p} \coqdocvar{q}) (\coqdocvar{F} \coqdocvar{P} \coqdocvar{Q}).\coqdoceol
}

\subsection{Soundness of the New Tnum Multiplication}

Following Definition~\ref{def:soundness-for-proofs}, we need to show
that \lstinline{tnum_mul(P, Q)} contains all $\mathbf{x} \times
\mathbf{y}$ such that $\mathbf{x} \in \gamma(\mathbf{P})$ and
$\mathbf{y} \in \gamma(\mathbf{Q})$. However, proving this directly
against the multiplication of \lstinline{Z} in \rocq{} is infeasible
(because the definition of \lstinline{Z} multiplication has not enough
resemblance to that of \lstinline{tnum_mul}). Thus, we define a new
function called \lstinline{uint_mul} for the multiplication of bit
vectors based on the shift-and-add approach. The soundness proof of
\lstinline{tnum_mul} now splits into the following:

\begin{enumerate}
\item The correctness of \lstinline{uint_mul} against the
  \rocq{}-provided \lstinline{Z} multiplication
\item The soundness of \lstinline{tnum_mul} against \lstinline{uint_mul}
\end{enumerate}

We show these by proving the following lemmas in \rocq{}:

\rocqlisting{\coqdockw{Lemma} \coqdocvar{uint\_mul\_correct} : \coqdockw{\ensuremath{\forall}} \coqdocvar{n} (\coqdocvar{a} \coqdocvar{b} : \coqdocvar{uint} (\coqdocvar{S} \coqdocvar{n})),\coqdoceol
\coqdocindent{3.00em}
\coqdocvar{uint\_denote} (\coqdocvar{uint\_mul} \coqdocvar{a} \coqdocvar{b}) =\coqdoceol
\coqdocindent{4.00em}
(\coqdocvar{uint\_denote} \coqdocvar{a} \ensuremath{\times} \coqdocvar{uint\_denote} \coqdocvar{b}) \coqdocvar{mod} \coqdocvar{modulus} (\coqdocvar{S} \coqdocvar{n}).\coqdoceol
}

(where \lstinline{uint_denote} gives the value of a \lstinline{uint} in
\lstinline{Z})

\rocqlisting{\coqdockw{Lemma} \coqdocvar{tnum\_mul\_sound} (\coqdocvar{n} : \coqdocvar{nat}) : \coqdocvar{sound2} (\coqdocvar{S} \coqdocvar{n}) \coqdocvar{uint\_mul} \coqdocvar{tnum\_mul}.\coqdoceol
}

As a prerequisite, we also prove the correctness of binary and tnum
shifts that we define and use in \lstinline{uint_mul} and
\lstinline{tnum_mul}. We omit such details for brevity.

\subsubsection{Proof Overview.}
The while loop in \lstinline{tnum_mul} is encoded as a recursive
function \lstinline{tnum_mul_loop(A, B, C)}. The additional argument
$C$ is the partial accumulator. $A$ is of length $m$ and both $B$ and
$C$ are of length $n$ (the actual bit width that \lstinline{tnum_mul}
operates on). $m$ decreases by one with each recursive call (i.e., the
while loop iteration where $A$ is right-shifted). Given
$\ingamma{a}{A}$, $\ingamma{b}{B}$, and $\ingamma{c}{C}$, the goal is
to show \lstinline{uint_mul_loop(a, b, c)} $\in$
\lstinline{tnum_mul_loop(A, B, C)}.

The proof progresses by induction on $m$, the length of $A$. The base
case is trivial, and for the induction step, we show that
$\ingamma{a'}{A'}$, $\ingamma{b'}{B'}$, and $\ingamma{c'}{C'}$ (where
the primed variables are the updated multiplier, multiplicand, and
partial accumulator calculated for the recursive call). This enables
us to apply the induction hypothesis and discharge the main
goal. $\ingamma{a'}{A'}$ and $\ingamma{b'}{B'}$ are due to the
soundness of tnum shifts. $\ingamma{c'}{C'}$ is tricky because
calculating the updated partial accumulator involves branching on
$A.m[0]$ and a union operation in \lstinline{tnum_mul_loop}, which are
not present in \lstinline{uint_mul_loop}. However, the proof is
straightforward thanks to the soundness of \lstinline{tnum_add} and
\lstinline{tnum_union} proved separately.

For a more detailed explanation, see Appendix~\ref{apx:proof-details}.

\subsection{Soundness and Optimality of Tnum Union}
\label{sec:tnum-union-sound}

The function \lstinline{tnum_mul} uses \lstinline{tnum_union} as a
subroutine. We define \lstinline{tnum_union} following
Algorithm~\ref{lst:algo-tnum-mul-new} which achieves the union
operation defined in Definition~\ref{def:tnum-union} without actually
computing the concrete sets represented by the input tnums. The
algorithm employs only a couple of bitwise operations, based on the
idea that each bit in the result can be computed depending only on the
input bits in the same position:

\begin{itemize}
\item When the $i$-th trits in both operands are the same certain
  value, the result also has the same certain value.
\item When either operand has uncertainty at the $i$-th position, or
  when the operands have certain but disagreeing values at the $i$-th
  position, the result has uncertainty at the $i$-th position.
\end{itemize}

While simple, the above reasoning or the exact use of bitwise
operations could be easily wrong. Hence we show the soundness of
\lstinline{tnum_union} by proving the following lemma, which states
that both $\mathbf{P}$ and $\mathbf{Q}$ are subsets of $\mathbf{P}
\cup \mathbf{Q}$:

\rocqlisting{\coqdockw{Lemma} \coqdocvar{tnum\_union\_sound} \{\coqdocvar{SIZE}\} (\coqdocvar{P} \coqdocvar{Q} : \coqdocvar{tnum.t} \coqdocvar{SIZE}) :\coqdoceol
\coqdocindent{2.00em}
\coqdocvar{tnum.wellformed} \coqdocvar{P} \ensuremath{\rightarrow} \coqdocvar{tnum.wellformed} \coqdocvar{Q} \ensuremath{\rightarrow}\coqdoceol
\coqdocindent{2.00em}
\coqdockw{let} \coqdocvar{U} := \coqdocvar{tnum\_union} \coqdocvar{P} \coqdocvar{Q} \coqdoctac{in}\coqdoceol
\coqdocindent{2.00em}
\coqdocvar{tnum.wellformed} \coqdocvar{U} \ensuremath{\land} \coqdocvar{subset} \coqdocvar{P} \coqdocvar{U} \ensuremath{\land} \coqdocvar{subset} \coqdocvar{Q} \coqdocvar{U}.\coqdoceol
}

This is how we define the subset relation of tnums:

\rocqlisting{\coqdockw{Definition} \coqdocvar{subset} \{\coqdocvar{SIZE}\} (\coqdocvar{P} \coqdocvar{Q} : \coqdocvar{tnum.t} \coqdocvar{SIZE}) :=\coqdoceol
\coqdocindent{1.00em}
\coqdockw{\ensuremath{\forall}} \coqdocvar{x}, \coqdocvar{ingamma} \coqdocvar{x} \coqdocvar{P} \ensuremath{\rightarrow} \coqdocvar{ingamma} \coqdocvar{x} \coqdocvar{Q}.\coqdoceol
}

The proof progresses by a case analysis on the $i$-th value/mask bits
of the operands after fixing a bit position $i$. It is fairly
straightforward because tnum union does not involve carries or loops.

\subsubsection{Optimality of \lstinline{tnum_union}.}
Our approach to proving the optimality of \lstinline{tnum_union} is to
show that any uncertainty in \lstinline{tnum_union(P, Q)} for any
$\mathbf{P}$, $\mathbf{Q}$ is the minimum required to represent their
concrete union $\gamma(\mathbf{P}) \cup \gamma(\mathbf{Q})$ as a
tnum. More specifically, if \lstinline{tnum_union(P, Q)} has
uncertainty at $i$-th bit for some position $i$, the concrete union
contains two elements such that their $i$-th bits have distinct
values. We define the concrete union and the optimality lemma in
\rocq{} as follows:

\rocqlisting{\coqdockw{Definition} \coqdocvar{concrete\_union\_element} \{\coqdocvar{SIZE}\} (\coqdocvar{P} \coqdocvar{Q} : \coqdocvar{tnum.t} \coqdocvar{SIZE}) :=\coqdoceol
\coqdocindent{2.00em}
\{ \coqdocvar{x} \ensuremath{|} \coqdocvar{ingamma} \coqdocvar{x} \coqdocvar{P} \ensuremath{\lor} \coqdocvar{ingamma} \coqdocvar{x} \coqdocvar{Q} \}.\coqdoceol
}

\rocqlisting{\coqdockw{Lemma} \coqdocvar{tnum\_union\_optimal} \{\coqdocvar{SIZE}\} (\coqdocvar{P} \coqdocvar{Q} : \coqdocvar{tnum.t} \coqdocvar{SIZE}) :\coqdoceol
\coqdocindent{2.00em}
\coqdocvar{tnum.wellformed} \coqdocvar{P} \ensuremath{\rightarrow} \coqdocvar{tnum.wellformed} \coqdocvar{Q} \ensuremath{\rightarrow}\coqdoceol
\coqdocindent{2.00em}
\coqdockw{\ensuremath{\forall}} \coqdocvar{i}, \coqdocvar{i} \texttt{<} \coqdocvar{SIZE} \ensuremath{\rightarrow}\coqdoceol
\coqdocindent{3.00em}
\coqdocvar{tnum.ith\_m} (\coqdocvar{tnum\_union} \coqdocvar{P} \coqdocvar{Q}) \coqdocvar{i} = \coqdocvar{true} \ensuremath{\rightarrow}\coqdoceol
\coqdocindent{3.00em}
\coqdoctac{\ensuremath{\exists}} (\coqdocvar{x} \coqdocvar{y} : \coqdocvar{concrete\_union\_element} \coqdocvar{P} \coqdocvar{Q}),\coqdoceol
\coqdocindent{4.00em}
\coqdocvar{uint\_ith} (\coqdocvar{proj1\_sig} \coqdocvar{x}) \coqdocvar{i} \ensuremath{\not=} \coqdocvar{uint\_ith} (\coqdocvar{proj1\_sig} \coqdocvar{y}) \coqdocvar{i}.\coqdoceol
}

The proof progresses by a case analysis on the $i$-th mask bits of the
operands after fixing a bit position $i$. The following table lists
the cases and choices of $x$ and $y$ that discharge the goal:

\begin{center}
  \begin{tabular}{ccll}
    \toprule
    $P.[i]$ & $Q.m[i]$ & $x$ & $y$ \\
    \midrule
    0 & 0 & $P.v$ & $Q.v$ \\
    0 & 1 & $Q.v$ & $Q.v \mathbin{|} Q.m$ \\
    1 & 0 & $P.v$ & $P.v \mathbin{|} P.m$ \\
    1 & 1 & $P.v$ & $P.v \mathbin{|} P.m$
  \end{tabular}
\end{center}

\subsection{Soundness and Optimality of Tnum Addition}
\label{sec:tnum-add-rocq}

\todo{cross-check: how it differs from previous work.}

The pen-and-paper proofs from the previous work~\cite{hariTnum} for
the soundness and optimality of \lstinline{tnum_add} focus on
``identifying which carry-in bit positions vary across different
concrete additions.'' In essence, they show the equivalence of the
formula for concrete additions and \lstinline{tnum_add}, concluding
\lstinline{tnum_add} to be both sound and optimal. Our proofs in
\rocq{}, on the other hand, show the soundness and optimality of
\lstinline{tnum_add} separately by proving direct
lemmas. \lstinline{tnum_add} is encoded in \rocq{} as a direct
translation of Listing~\ref{lst:tnum-add} using bit operations and
binary addition defined on \lstinline{uint_add}.

\subsubsection{Soundness of \lstinline{tnum_add}.}

Following Definition~\ref{def:soundness-for-proofs}, we show the
soundness as well as the well-formedness of \lstinline{tnum_add} by
proving the lemma given below:

\rocqlisting{\coqdockw{Lemma} \coqdocvar{tnum\_add\_sound} \{\coqdocvar{SIZE}\} : \coqdocvar{sound2} \coqdocvar{SIZE} \coqdocvar{uint\_add} \coqdocvar{tnum\_add}.\coqdoceol
}

The proof is essentially a deep case analysis. Our key strategy was to
fill the context with several constraints implicit in tnum addition,
helping \rocq{} eliminate absurd cases during
\lstinline{destruct}. For example, if the incoming carry bit at
position $i$ is $0$ for both $P.m + Q.m$ and $(P.v + P.v) + (Q.m +
Q.v)$, then the $i$-th bit in the concrete sum is the same as the
$i$-th value bit in the abstract sum. These expressions appear in the
definition of \lstinline{tnum_add} as \lstinline{sm} and
\lstinline{sv + sm} respectively (see Listing~\ref{lst:tnum-add}).

\subsubsection{Optimality of \lstinline{tnum_add}.}

Like in the case of the optimality of \lstinline{tnum_union}, we prove
the optimality of \lstinline{tnum_add} by showing that if there is
uncertainty at the $i$-th bit of \lstinline{tnum_add(P, Q)} for some
$\mathbf{P}$, $\mathbf{Q}$, $i$, then there exist a pair of concrete
sums that have distinct bits at position $i$:

\rocqlisting{\coqdockw{Lemma} \coqdocvar{tnum\_add\_optimal} : \coqdockw{\ensuremath{\forall}} [\coqdocvar{SIZE}] (\coqdocvar{P} \coqdocvar{Q} : \coqdocvar{tnum.t} \coqdocvar{SIZE}) \coqdocvar{i} (\coqdocvar{hidx} : \coqdocvar{i} \texttt{<} \coqdocvar{SIZE}),\coqdoceol
\coqdocindent{3.00em}
\coqdocvar{tnum.wellformed} \coqdocvar{P} \ensuremath{\rightarrow} \coqdocvar{tnum.wellformed} \coqdocvar{Q} \ensuremath{\rightarrow}\coqdoceol
\coqdocindent{3.00em}
\coqdocvar{tnum.ith\_m} (\coqdocvar{tnum\_add} \coqdocvar{P} \coqdocvar{Q}) \coqdocvar{i} = \coqdocvar{one} \ensuremath{\rightarrow}\coqdoceol
\coqdocindent{3.00em}
\coqdoctac{\ensuremath{\exists}} \coqdocvar{p} \coqdocvar{q} \coqdocvar{p'} \coqdocvar{q'},\coqdoceol
\coqdocindent{4.00em}
\coqdocvar{ingamma} \coqdocvar{p} \coqdocvar{P} \ensuremath{\land} \coqdocvar{ingamma} \coqdocvar{q} \coqdocvar{Q} \ensuremath{\land}\coqdoceol
\coqdocindent{5.00em}
\coqdocvar{ingamma} \coqdocvar{p'} \coqdocvar{P} \ensuremath{\land} \coqdocvar{ingamma} \coqdocvar{q'} \coqdocvar{Q} \ensuremath{\land}\coqdoceol
\coqdocindent{5.00em}
\coqdocvar{uint\_ith} (\coqdocvar{uint\_add} \coqdocvar{p} \coqdocvar{q}) \coqdocvar{i} \ensuremath{\not=} \coqdocvar{uint\_ith} (\coqdocvar{uint\_add} \coqdocvar{p'} \coqdocvar{q'}) \coqdocvar{i}.\coqdoceol
}

(It is implicit that either of $p \ne p'$ or $q \ne q'$ should hold,
but both need not.)

There are exactly three cases by which the $i$-th bit of
\lstinline{tnum_add(P, Q)} can be $1$. Either $P.m[i] = 1$, $Q.m[i] =
1$, or $\chi.m[i] = 1$ (see Listing~\ref{lst:tnum-add}). We consider
these three cases individually to prove the above lemma. The first two
cases are straightforward. When $\chi.m[i] = 1$, we show that the
following assignments would satisfy our goal:

\begin{align*}
  p &= P.v \\
  q &= Q.v \\
  p' &= P.v \mathbin{|} P.m \\
  q' &= Q.v \mathbin{|} Q.m
\end{align*}

One key insight used in the proof is that the addition of the value
and mask fields of a well-formed tnum is equivalent to their bitwise
OR. The details can be found in Appendix~\ref{apx:proof-details}.

\todo{include tnum\_union\_optimal\_by\_subset?}

\section{Experimental Evaluation}

We evaluate the \emph{optimality}, \emph{precision}, and
\emph{performance} of our \lstinline{tnum_mul} (upstreamed in Linux
6.18) by experimental means and compare the results with that of the
previous one~\cite{hariTnum}. We also investigate the use of
multiplication in real-world \ebpf{} programs. Throughout this
section, we refer to our \lstinline{tnum_mul} as
\lstinline{our_tnum_mul} and the previous one as
\lstinline{prv_tnum_mul}.

\subsection{Optimality}

To estimate the overall optimality of \lstinline{tnum_mul} at a given
bit width, we pick all possible tnum pairs $(\mathbf{P}, \mathbf{Q})$
and count the instances where \lstinline{tnum_mul(P, Q)} is the same
as $\alpha(\gamma(\mathbf{P}) \times \gamma(\mathbf{Q}))$ (recall
Definition~\ref{def:optimal}). Table~\ref{tab:optimality} shows that
\lstinline{our_tnum_mul} performed better at all bit widths tested. We
could not test beyond eight bits as it was computationally infeasible
(at bit width $n$, we need to find the entire concrete product set of
$3^n \times 3^n$ tnum pairs and then abstract it).

\begin{table}
  \begin{center}
    \begin{tabular}{rrr}
      \toprule
      Bit Width & Ours Is Optimal & Prv. Is Optimal \\
      \midrule
      2 & 97.53\% & 96.30\% \\
4 & 90.63\% & 68.62\% \\
8 & 76.08\% & 21.82\% \\
\bottomrule

    \end{tabular}
  \end{center}

  \caption{Percentage of instances where \lstinline{our_tnum_mul} and
    \lstinline{prv_tnum_mul} yielded optimal results.}
  \label{tab:optimality}
\end{table}

As seen in Table~\ref{tab:optimality}, both \lstinline{our_tnum_mul}
and \lstinline{prv_tnum_mul} become less optimal as the bit width
increases. However, the rate of \lstinline{our_tnum_mul} becoming less
optimal is much less than that of \lstinline{prv_tnum_mul}.

\subsection{Relative Precision}

Finding the relative precision of two tnum multiplication algorithms
is easier than checking if a particular algorithm yields truly optimal
results. Hence, at higher bit widths, this serves as a practical
measure for how optimal a version of \lstinline{tnum_mul} is compared
to another.

Given two sound tnum multiplication functions $f_1$ and $f_2$, we say
$f_1$ yields a more precise product than $f_2$ for operands
$\mathbf{P}, \mathbf{Q}$ if $f_1(\mathbf{P}, \mathbf{Q}) \tle
f_2(\mathbf{P}, \mathbf{Q})$. Checking this does not require finding
the optimal product and can be quickly done with the help of the
function \lstinline{tnum_in} available in the Linux kernel (which is
implemented using a few bitwise
operations)~\cite{linux-tnum-h, linux-tnum-c}.

To find the overall precision of \lstinline{our_tnum_mul} over a
sample set, we check for how many pairs of operands it yielded a
precision better than, similar to, or worse than
\lstinline{prv_tnum_mul}. Incomparable cases are in which the result
of one multiplication algorithm does not represent a subset of the
other.

We checked the precision of both \lstinline{our_tnum_mul} and
\lstinline{prv_tnum_mul} for all possible tnum pairs up to 8 bits, and
the results are given in Table~\ref{tab:precision-exhaust}.

\begin{table}
  \begin{center}
    \begin{tabular}{rrrrr}
      \toprule
      Bit Width & Ours Is Better & Both the Same & Ours Is Worse & Incomparable \\
      \midrule
      2 & 1.23\% & 98.77\% & 0.00\% & 0.00\% \\
4 & 25.54\% & 72.76\% & 1.60\% & 0.09\% \\
8 & 69.84\% & 26.45\% & 3.30\% & 0.41\% \\
\bottomrule

    \end{tabular}
  \end{center}

  \caption{Percentage of instances where \lstinline{our_tnum_mul} had
    a precision better than, same as, or worse than
    \lstinline{prv_tnum_mul} at lower bit widths.}
  \label{tab:precision-exhaust}
\end{table}

At higher bits, due to the infeasibility of considering all tnum
pairs, we used 10 million random pairs
(results given in Table~\ref{tab:precision-rand}). The exact number of incomparable
cases are listed in the table because the percentage becomes zero
after rounding.

\begin{table}
  \begin{center}
    \begin{tabular}{rrrrrr}
      \toprule
      & & & & \multicolumn{2}{c}{Incomparable} \\
      \cmidrule(r){5-6}
      Bit Width & Ours Is Better & Both the Same & Ours Is Worse & \% & Count / 10M \\
      \midrule
      16 & 93.55\% & 6.23\% & 0.18\% & 0.04\% & 4377 \\
24 & 95.41\% & 4.37\% & 0.17\% & 0.04\% & 4389 \\
31 & 95.59\% & 4.20\% & 0.17\% & 0.04\% & 4479 \\
32 & 43.41\% & 56.38\% & 0.19\% & 0.02\% & 2073 \\
48 & 0.54\% & 99.42\% & 0.04\% & 0.00\% & 235 \\
64 & 0.54\% & 99.42\% & 0.04\% & 0.00\% & 235 \\
\bottomrule

    \end{tabular}
  \end{center}

  \caption{Percentage of instances where \lstinline{our_tnum_mul} had
    a precision better than, same as, or worse than
    \lstinline{prv_tnum_mul} at higher bit widths.}
  \label{tab:precision-rand}
\end{table}

Data shows that \lstinline{our_tnum_mul} yields more precise results
compared to the previous one for all bit widths tested.  Our results
are in alignment with the independent benchmarks performed by others
as part of the review process of our kernel
patch~\cite{eddy-benchmark, hari-benchmark}. As shown in the results,
\lstinline{our_tnum_mul} performs worse than \lstinline{prv_tnum_mul}
for some inputs. This is because the accumulator contains both value
and mask bits (in \lstinline{prv_tnum_mul}, the accumulator contains
mask bits only). This increases the imprecision of
\lstinline{our_tnum_mul} when \lstinline{tnum_add} is performed in the
intermediate steps. Running both algorithms on the example (a =
$0111$, b = $01\mu1$) would make this clear. However, we found that
incorporating the \emph{value-mask decomposition} technique from the
previous work in \lstinline{our_tnum_mul} does not improve the
precision further.

\ifdefined\doplot
To get a quick idea of the trend in the improvement of precision, data
from Table~\ref{tab:precision-exhaust} and
Table~\ref{tab:precision-rand} is plotted in
Figure~\ref{fig:precision}. Incomparable cases are omitted in the plot
because the number of such instances are negligible ($0.41\%$ at 8
bits and $<0.1\%$ in all other bit widths tested).

\begin{figure}
  \begin{center}
\begin{tikzpicture}
  \begin{axis}[
      legend pos=outer north east,
      xmode=log,
      log basis x={2},
      xlabel=Bit Width,
      ylabel=Percentage,
    ]
    \pgfplotsset{cycle list shift=10}
    \addplot gnuplot [
      raw gnuplot,
      table/x index=0,
      table/y index=1,
    ] {plot 'precision-eval/tab-gen-by-mkgnuplot.txt' index 0};
    \addplot gnuplot [raw gnuplot] {plot 'precision-eval/tab-gen-by-mkgnuplot.txt' index 0};
    \pgfplotsset{cycle list shift=6}
    \addplot gnuplot [raw gnuplot] {plot 'precision-eval/tab-gen-by-mkgnuplot.txt' index 1};
    \pgfplotsset{cycle list shift=9}
    \addplot gnuplot [raw gnuplot] {plot 'precision-eval/tab-gen-by-mkgnuplot.txt' index 2};
    \legend{new is more precise, both the same, new is less precise}
  \end{axis}
\end{tikzpicture}
\caption{Percentage of instances where the new algorithm had a
  precision better than, same as, or worse than the previous one at
  different bit widths.}
\label{fig:precision}
\end{center}
\end{figure}
\fi

\subsubsection{Precision Anomaly Beyond 31 Bits.}

\label{sec:precision-anomaly}

As seen from Table~\ref{tab:precision-exhaust} and
Table~\ref{tab:precision-rand},
\lstinline{our_tnum_mul} progressively does better than
\lstinline{prv_tnum_mul} until the 31-bit mark.  The seemingly
unexpected drop in advantage beyond this point is an artifact of using
64-bit integers in the implementations. Beyond 31 bits, the result of
a multiplication can overflow 64 bits. A similar anomaly can be
observed at the 16-bit mark when using 32-bit integers to implement
tnum algorithm.

\subsection{Precision at Higher Bits}

The precision anomaly mentioned above lead us to the following
hypothesis.

\begin{nonumhypo}
  The precision advantage that our algorithm has over the previous one
  is towards the most significant bits.
\end{nonumhypo}

Precision at higher bits is important because their positional values
are higher. For example, an array index being $\mu0000001$ (imprecise
at a high bit) is more likely to result in a false positive (a good
program being rejected) during bounds checking than $000000\mu1$.

To verify our hypothesis, we compared the results of
\lstinline{our_tnum_mul} and \lstinline{prv_tnum_mul} for all tnum
pairs up to 8 bits and for 10 million random tnum pairs at higher
bit widths. The exact check we did is to see which version of
\lstinline{tnum_mul} returned a mask that is numerically smaller (this
would mean the $\mu$ trits are at lower bits, or in other words, the
result is more precise at higher bits). The data is summarized in
Table~\ref{tab:precision-higherbits}, which confirms our hypothesis.

\begin{table}
  \begin{center}
    \begin{tabular}{rrr}
      \toprule
      Bit Width & Out of All Cases & Out of Incomparable Cases \\
      \midrule
      2 & 1.23\% & -nan\%
\\
4 & 25.54\% & 0.00\%
\\
8 & 70.07\% & 55.92\%
\\
16 & 93.59\% & 97.08\%
\\
24 & 95.45\% & 97.74\%
\\
31 & 95.63\% & 97.48\%
\\
32 & 43.43\% & 88.28\%
\\
48 & 0.54\% & 1.28\%
\\
64 & 0.54\% & 1.28\%
\\
\bottomrule

    \end{tabular}
  \end{center}

  \caption{Percentage of instances where \lstinline{our_tnum_mul}
    yielded more precision at higher bits compared to
    \lstinline{prv_tnum_mul}, at different bit widths.}
  \label{tab:precision-higherbits}
\end{table}

\subsection{Multiplication in Real-world eBPF Programs}

\label{sec:mul-usage-stats}

To quantify the impact of improvements to tnum multiplication, we
analyzed the disassembly of 251 eBPF object files from the sample
set~\cite{vbpf-samples} used in PREVAIL~\cite{prevail}. The set
includes programs from major projects like Cilium~\cite{cilium}. There
were 8854 instances of \lstinline{*=}, implying real-world eBPF
programs use multiplication significantly. Out of these, 2167 were
64-bit multiplication and 6687 were 32-bit. This is noteworthy because
at 32 bits, our algorithm performs better than the previous one in
43.41\% of cases even with the precision anomaly mentioned in
Section~\ref{sec:precision-anomaly}.

\subsection{Performance}

Since our algorithm improves the precision significantly, it makes
sense to ask if this improvement comes with a performance overhead.
While both algorithms seem to have the same complexity and termination
conditions, we take measurements to get the real-world performance.

We ran both versions of \lstinline{tnum_mul} for 10 million times in a
single thread and took the measurement using
\lstinline{clock_gettime(CLOCK_THREAD_CPUTIME_ID)}. We did this in the
userspace, which does not make any difference in the case of
\lstinline{tnum_mul}. We compiled the program for x86\_64 using
\lstinline{gcc -O2} (the same optimization level the Linux kernel is
compiled at~\cite{linux-makefile}) and ran it on Intel Core i5-6300U
CPU.

For 10 million random pairs, on average, the previous algorithm
finished in 4.52 seconds while our algorithm took 5.19 seconds. While
not considerable, this shows our algorithm has a performance overhead.

We used \emph{callgrind}~\cite{callgrind} to get the cycle estimation,
which showed \lstinline{our_tnum_mul} took 1352 cycles per call while
\lstinline{prv_tnum_mul} took 875 cycles. While we expect our use of
\lstinline{tnum_union} to cause some overhead, we couldn't get the
cycle estimation for the same since the compiler inlined
\lstinline{tnum_union} inside \lstinline{tnum_mul} (observed using
\emph{objdump}).

\section{Related Work}

The research community has shown deep interest in \ebpf{}
verification.  Previously, bugs have been identified in the range
analysis part of the Linux kernel \ebpf{} verifier, and works exist
that apply formal methods to rectify such
issues~\cite{shachnai-fixing-ebpf-abs-ops, agni, prevail}. Another
component of the eBPF ecosystem that benefited from formal methods is
its JIT compiler in the Linux kernel~\cite{osdiJit}.

The first attempt at formalizing the Linux tnum arithmetic is by
\citeauthor{hariTnum}~\cite{hariTnum}, where the authors provide
bounded-bitwidth machine-checked proofs and unbounded-bitwidth
pen-and-paper proofs for the soundness and optimality of
\lstinline{tnum_add} and \lstinline{tnum_sub}. They also provide an
improved algorithm for multiplication with soundness proof. Their
implementation was accepted to the Linux kernel and remained there
until our patch~\cite{nan-libbpf-patch} replaced it in Linux 6.18.

\todo{some rephrasing needed here in the commented part.}

The \emph{KnownBits} domain used in LLVM~\cite{llvm-kb-cpp} is related
to tnum from Linux in spirit, but instead of mask and value, it has
two fields \lstinline{Zero} and \lstinline{One}, keeping track of
positions known to be $0$ and $1$ respectively~\cite{llvm-kb-h}. This
is similar to the the \emph{bitfield domain} $\mathcal{D}^{\sharp}_b$
presented by Min{\'e}~\cite{mine-bitfield}.

\section{Conclusion}

Tnum is an important abstract domain that is used as part of the Linux
kernel \ebpf{} verifier. Closely related variants can be found in
other software as well. This paper presents a novel algorithm
(upstreamed to the Linux kernel) for tnum multiplication that is more
precise and equally efficient compared to the previous one. We prove
the soundness of the same in \rocq{}, along with soundness and
optimality proofs for tnum addition and union. Ours is the first work
to provide unbounded machine-checked proofs for tnum arithmetic.

\ifdefined\doubleblind
\else
\begin{credits}
  \subsubsection{\ackname}
  We thank Eduard Zingerman and Harishankar Vishwanathan for
  reviewing, testing, and independently benchmarking our Linux kernel
  patch containing the new \lstinline{tnum_mul} code. We also thank
  Alexei Starovoitov, Daniel Borkmann, and Andrii Nakryiko for their
  valuable comments and participation in the process involved in
  getting the patch merged.
\end{credits}
\fi

\printbibliography

\appendix
\section{New Algorithm and the \ebpf{} Verifier}
\label{apx:test-program}

Following is an \ebpf{} Assembly program included in the Linux kernel
test suite as part of our patch. It demonstrates how the precision of
the new tnum multiplication algorithm helps the verifier avoid a false
positive (i.e., a good program being rejected).

\begin{lstlisting}
	call %[bpf_get_prandom_u32];
	r0 &= 0x2;
	r0 |= 0x1;
	r0 *= 0x3;
	r0 &= 0x4;
	if r0 != 0 goto l0_%=;
	r0 = 0;
	goto l1_%=;
l0_%=:
	r0 = 1;
l1_%=:
\end{lstlisting}

The \ebpf{} runtime requires that every \ebpf{} program terminates
with the register \lstinline{r0} set to \lstinline{0}. In the above
program, the branch \lstinline{l0} breaks this rule. However, this
branch is never taken, and whether the program is accepted or not
depends on the verifier's ability to figure this out.

In the first three lines, we force the verifier to be uncertain about
the value in \lstinline{r0} by doing an operation involving a random
value. After analyzing these, the verifier must be representing
\lstinline{r0} as $00\mu1$. The following operation \lstinline{r0 *= 0x3}
puts the verifier in the same situation illustrated in
Figure~\ref{fig:new-mul-example}, resulting in \lstinline{r0} being
$\mu\mu\mu1$ when the previous algorithm is used, and \lstinline{r0}
being $\mu0\mu1$ if our algorithm is used. The added precision in the
new algorithm helps the verifier realize that the illegal branch is
never taken, since the program takes this branch based on the result
of the multiplication.

\section{Detailed Overview of the Proofs}
\label{apx:proof-details}

\subsection{Convention}

This section uses upper case letters to denote tnums that have a
value-mask representation internally.  $\wfdomn$ is used to represent
the set of all possible well-formed (see
Definition~\ref{def:wellformed}) tnums in $n$-bit value-mask
representation.  For a tnum $P$, $\ithv{P}$ represents the $i$-th
value bit of $P$, and $\ithm{P}$ represents the $i$-th mask bit of $P$
(like in the Rocq proofs, this document uses lower numbers to
represent least significant bits, but 1-indexing is used instead of
0-indexing). For a bit vector $x$, $\seti{x}$ denotes the same bit
vector with its $i$-ith bit set to $1$.

\subsection{Background}

In this section, $i$ denotes a valid index (bit position) for the bit
vector it is used with. Details like bounds are omitted for brevity.

\begin{definition}
  \label{def:wellformed} (Well-formedness)
  $$\forall P \in \wfdomn, \ithm{P} = 1 \implies \ithv{P} = 0$$
\end{definition}

\begin{definition}
  \label{def:ingamma} (Membership)
  $$\forall p \in \{0, 1\}^{n}, P \in \wfdomn : \ingamma{p}{P} \iff (\ithm{P} = 0 \implies \ith{p} = \ithv{P})$$
\end{definition}

\begin{lemma}
  \label{lem:ingamma_value}
  $\forall P \in \wfdomn, \ingamma{\tv{P}}{P}$
\end{lemma}
\begin{proof}
  Follows from Definition~\ref{def:ingamma}.
\end{proof}

\begin{lemma}
  \label{lem:ingamma_value_bitor_mask}
  $\forall P \in \wfdomn, \ingamma{\orvm{P}}{P}$
\end{lemma}
\begin{proof}
  Follows from Definition~\ref{def:ingamma}.
\end{proof}

\begin{lemma}
  \label{lem:ingamma_set_at_mask}
  Consider any $P \in \wfdomn$. If $\ithm{P} = 1$, then $\ingamma{\seti{\tv{P}}}{P}$.
\end{lemma}
\begin{proof}
  Follows from Definition~\ref{def:ingamma}. (Note that
  $\seti{\tv{P}}$ is the same as $\tv{P}$, except at position $i$. At
  position $i$, the mask is set, so the value bit is free to be either
  $0$ or $1$ without losing membership.)
\end{proof}

\begin{lemma}
  \label{lem:tnum_add_v_m_is_or}
  $\forall P \in \wfdomn, \tv{P} + \tm{Q} = \tv{P} \mathbin{|} \tm{Q}$
\end{lemma}
\begin{proof}
  Due to well-formedness, $\ithv{P}$ and $\ithm{P}$ cannot be both $1$
  for any $i$. For such inputs, addition decomposes to bitwise OR.
\end{proof}

\subsection{Optimality of \lstinline{tnum_add}}

Informally speaking, an abstract arithmetic operation is optimal if it
has no more uncertainty than there is in the concrete result set.

In the case of tnum addition, this means if the abstract sum of two
tnums have $\mu$ at some position $i$, there should be at least one
pair of concrete sums that differ at $i$. We formally state and prove
this in the following lemma, showing \lstinline{tnum_add} from the
Linux kernel is optimal.

\begin{lemma}
  Consider any $P, Q \in \wfdomn$, and a position $i \le n$. If $\ithm{\tadd{P}{Q}} = 1$, then
  \[
  \exists \ingamma{p, p'}{P}, \exists \ingamma{q, q'}{Q} : \ith{(p + q)} \neq \ith{(p' + q')}.
  \]
\end{lemma}

\begin{proof}
  There are exactly three cases by which $\ithm{\tadd{P}{Q}}$ can be
  $1$. Either $\ithm{P} = 1$, $\ithm{Q} = 1$, or $\ithm{\chi} = 1$.
  We consider these three cases individually and show that there exist
  $\ingamma{p, p'}{P}$ and $\ingamma{q, q'}{Q}$ such that $\ith{(p +
    q)} \neq \ith{(p' + q')}$. (Note that either of $p \neq p'$ or $q
  \neq q'$ must hold, but both need not be true.)

  \Case{1}{$\ithm{P} = 1$}. Taking $p = \tv{P}$, $p' = \seti{\tv{P}}$,
  and $q = q' = \tv{Q}$ satisfies our goal. Their membership follows
  from Lemma~\ref{lem:ingamma_value} and
  Lemma~\ref{lem:ingamma_set_at_mask}. We know that $\ithv{P} = 0$
  since {$\ithm{P} = 1$}. Hence $p$ and $p'$ differ at (and only at)
  position $i$, and $\ith{(p + q)} \neq \ith{(p' + q')}$.

  \Case{2}{$\ithm{Q} = 1$}. Similar to the previous case, we take $q =
  \tv{Q}$, $q' = \seti{\tv{Q}}$, and $p = p' = \tv{P}$.

  \Case{3}{$\ithm{\chi} = 1$}. Recall that $\chi = (\sv{P}{Q} +
  \sm{P}{Q}) \oplus \sv{P}{Q}$. Since we know that $\ithm{\chi} = 1$,
  we have two possibilities here:

  \Case{3.1}{$\ith{(\sv{P}{Q} + \sm{P}{Q})} = 1, \ithv{P} = \ithm{P} = \ith{\sv{P}{Q}} = 0$}.

  After a basic rearranging and application of
  Lemma~\ref{lem:tnum_add_v_m_is_or}, $\ith{(\orvm{P} + \orvm{Q})} =
  1$. We know that $\ingamma{\orvm{P}}{P}$ and $\ingamma{\orvm{Q}}{Q}$
  (by Lemma~\ref{lem:ingamma_value_bitor_mask}). Hence we can pick
  $\ingamma{p, p'}{P}$ and $\ingamma{q, q'}{Q}$ in the following
  manner to satisfy $\ith{(p + q)} \neq \ith{(p' + q')}$:

  \begin{align*}
    p &= \tv{P} \\
    q &= \tv{Q} \\
    p' &= \orvm{P} \\
    q' &= \orvm{Q}
  \end{align*}

  \Case{3.2}{$\ith{\sv{P}{Q}} = 1, \ith{(\sv{P}{Q} + \sm{P}{Q})} =
    \ithv{P} = \ithm{P} = 0$}. Same assignments for $p$, $p'$, $q$ and
  $q'$ work in this case as well.
\end{proof}

As a side note, it is interesting to note that $\tv{P}$ is the
smallest concrete value and $\orvm{P}$ is the largest concrete value
represented by $P$, for any $P$.

\subsection{Optimality of \lstinline{tnum_union}}

The concrete union of two tnums $P$ and $Q$ is $\gamma(P) \cup
\gamma(Q)$. This can be rewritten like in the following:

\begin{definition}
  \[
  \text{concrete\_union}(P, Q) = \{ x \mathbin{|} \ingamma{x}{P} \lor \ingamma{x}{Q} \}
  \]
\end{definition}

\begin{lemma}
  Consider any $P, Q \in \wfdomn$, and a position $i \le n$. If $\ithm{\tunion{P}{Q}} = 1$, then
  \[
  \exists x, y \in \text{concrete\_union}(P, Q) : \ith{x} \neq \ith{y}.
  \]
\end{lemma}

\begin{proof}
  The proof progresses on a case analysis of $\ithm{P}$ and $\ithm{Q}$.

  \Case{1}{$\ithm{P} = \ithm{Q} = 0$}. Since we know that
  $\ithm{\tunion{P}{Q}} = 1$, by the definition of
  \lstinline{tnum_union}, $\ithv{P} \oplus \ithv{Q} = 1$.  This means
  $\ithv{P} \neq \ithv{Q}$, and since $\ingamma{\tv{P}}{P} \land
  \ingamma{\tv{Q}}{Q}$, we pick $x = \tv{P}$ and $y = \tv{Q}$,
  satisfying our goal.

  \Case{2}{$\ithm{P} = 0, \ithm{Q} = 1$}. In this case, $\gamma({Q})$
  itself has two elements that disagree at the $i$-th bit: $x =
  \tv{Q}$ and $y = \orvm{Q}$ (membership of the latter,
  $\ingamma{y}{Q}$, is shown by
  Lemma~\ref{lem:ingamma_value_bitor_mask}).

  \Case{3}{$\ithm{P} = 1, \ithm{Q} = 0$}. Similar to the previous
  case. We pick $x = \tv{P}$ and $y = \orvm{P}$.

  \Case{4}{$\ithm{P} = \ithm{Q} = 1$}. We can reuse the choices for
  $x$ and $y$ from either Case 2 or Case 3 in this case.
\end{proof}

\subsection{Soundness of \lstinline{tnum_mul}}

\newcommand{\bmulshrId}{\text{bvec\_mul\_loop}}
\newcommand{\bmulshrABC}{\bmulshr{a}{b}{c}}
\newcommand{\bmulshrABCp}{\bmulshr{a'}{b'}{c'}}
\newcommand{\bmulshr}[3]{$\bmulshrId{}(#1, #2, #3)$}

\newcommand{\tmulshrId}{\text{tnum\_mul\_loop}}
\newcommand{\tmulshr}[3]{$\tmulshrId{}(#1, #2, #3)$}
\newcommand{\tmulshrABC}{\tmulshr{A}{B}{C}}
\newcommand{\tmulshrABCp}{\tmulshr{A'}{B'}{C'}}
\newcommand{\tmulId}{\text{tnum\_mul}}
\newcommand{\tmulAB}{$\tmulId{}(A, B)$}

We model the while loop from our algorithm as a recursive function
\tmulshrABC{}, where $C$ is the partial accumulator. Note that $A$ is
a well-formed tnum of length $m \ge 0$, and the other two arguments
are well-formed tnums of length $n > 0$. \tmulAB{}, where $A$ and $B$
are well-formed tnums of equal length, can be defined as
\tmulshr{A}{B}{\text{TNUM}(0, 0)}.

Apart from taking the partial accumulator as an input argument,
\tmulshrId{} employs unpadded right-shift, reducing the length of the
multiplier $A$ by one (this does not affect its value compared to the
zero-padded right shift used in the actual algorithm). It is the
truncated multiplier that gets passed to the recursive call
corresponding to the next iteration of the while loop. This truncation
and the presence of the partial accumulator as an argument helps us
present a soundness proof using induction.

Now, the main goal of the proof becomes the following:

\begin{lemma}
  Consider a well-formed tnum $A$ and a bit vector $\ingamma{a}{A}$,
  both of length $m \ge 0$. Also consider well-formed tnums $B$, $C$,
  and bit vectors $\ingamma{b}{B}, \ingamma{c}{C}$, all of length $n >
  0$. Then

  \[
  \ingamma{\bmulshrId(a, b, c)}{\tmulshrId(A, B, C)}
  \]

  where \bmulshrId{}~\footnote{Proven to be correct against the
  built-in multiplication of natural numbers in Rocq.} is a bit-vector
  analog of \tmulshrId{}.
\end{lemma}

\begin{proof}
  The proof progresses by induction on $m$, the length of $A$.

  \basecase{$m = 0.$} For a zero-length multiplier, both the tnum and
  bit vector multiplication algorithms return the partial accumulator
  as the final product. This satisfies our goal due to the
  precondition $\ingamma{c}{C}$.

  \inductionStep{} When $m > 0$, \tmulshrABC{} is \tmulshrABCp{},
  where $C'$ is the updated accumulator obtained using tnum addition
  and union as needed, and $A'$ and $B'$ are the shifted versions of
  the multiplication operands. Similarly, \bmulshrABC{} becomes
  \bmulshrABCp{}. So, we need to show the membership of \bmulshrABCp{}
  in \tmulshrABCp{}. But this will follow from the induction
  hypothesis if we can show $\ingamma{a'}{A'}$, $\ingamma{b'}{B'}$,
  and $\ingamma{c'}{C'}$. The first two memberships follow from the
  soundness of tnum shifts. But $\ingamma{c'}{C'}$ needs discussion
  because the calculation of $C$ is not straightforward.

  There are only two cases to consider: $\lsb{\tv{A}}
  = 1$ and $\lsb{\tm{A}} = 1$.
  In the first case, $C' = \text{tnum\_add}(C, B)$, where the bit
  vector multiplication algorithm computes $c' = \text{bvec\_add}(c,
  b)$. $\ingamma{c'}{C'}$ due to the soundness of
  \text{tnum\_add}.

  \newcommand{\newCforMu}{\text{tnum\_union}(C, \text{tnum\_add}(C, B))}

  When $\lsb{\tm{A}} = 1$, the tnum multiplication algorithm computes
  $C' = \newCforMu{}$. Since $\lsb{a}$ can be either $0$ or $1$ in
  this case, we need to show the following, considering two different
  possible executions of \bmulshrId{}:

  \begin{enumerate}
  \item $\ingamma{c}{\newCforMu{}}$ (case: $\lsb{a} = 0$)
  \item $\ingamma{\text{bvec\_add}(c, b)}{\newCforMu{}}$ (case:
    $\lsb{a} = 1$)
  \end{enumerate}

  Both are satisfied by the memberships given in the preconditions,
  the soundness of \text{tnum\_add}, and the soundness of
  \text{tnum\_union}.

  Since we showed that $\ingamma{a'}{A'}$, $\ingamma{b'}{B'}$, and
  $\ingamma{c'}{C'}$, our original goal follows from the induction
  hypothesis.
\end{proof}

\end{document}